\documentclass{article}

\usepackage{amssymb}
\usepackage{amsmath}
\usepackage{amsfonts}
\usepackage{amsthm}

\usepackage[all]{xy}
\xyoption{color}
\xyoption{ps} 
\xyoption{dvips}
\usepackage{enumerate}
\usepackage{afterpage}
\usepackage{floatflt}
\usepackage{array}
\usepackage{multirow}
\usepackage{graphicx}
\usepackage{psboxit}
\usepackage{subfigure}
\usepackage[dvips, usenames]{color} 

\newcounter{network}

\newcommand{\uu}[1]{\ensuremath{\underline{#1}}}
\newcommand{\brac}[1]{\ensuremath{\left( #1 \right)}}
\newcommand{\rbrac}[1]{\ensuremath{\left[ #1 \right]}}

\def\R{\ensuremath{I\!\!R}}
\def\Rnn{\ensuremath{I\!\!R_{\geq 0}}}
\def\Rp{\ensuremath{\R_{> 0}}}
\def\cone{\ensuremath{\ker(Y\, I_a) \cap I\!\!R_{\geq 0}^r}}
\def\LAM{\ensuremath{\Lambda\brac{E}}}

\newcommand{\vel}[2]{\ensuremath{v( #1,\, #2 )}}

\DeclareMathOperator{\im}{im}
\DeclareMathOperator{\supp}{supp}
\DeclareMathOperator{\sign}{sign}
\DeclareMathOperator{\diag}{diag}

\DeclareMathOperator{\rank}{rank}
\DeclareMathOperator{\reac}{reac}

\newcommand{\sig}[1]{\ensuremath{\sign\brac{#1}}}

\newtheorem{defi}{Definition}

\newtheorem{lem}{Lemma}
\newtheorem{coro}{Corollary}

\newtheorem{remark}{Remark}

\newtheorem{ass}[defi]{Assumption}

\title{Multiple positive steady states in subnetworks defined by
  stoichiometric generators}
\author{Carsten Conradi\thanks{E-mail:  conradi@mpi-magdeburg.mpg.de} \\[0.2cm]
  Max-Planck-Institute Dynamics of Complex Technical
    Systems\\
    Sandtorstr. 1, 39106 Magdeburg, Germany
}
\begin{document}

\maketitle

\begin{abstract}
  In Systems Biology there is a growing interest in the question,
  whether  or not a given mathematical model can admit more than one
  steady state. As parameter values (like rate constants and total
  concentrations) are often unknown or subject to a very high
  uncertainty due to measurement errors and and difficult experimental
  conditions, one is often interested in the question, whether or not
  a given mathematical model can, for some conceivable parameter
  vector, exhibit multistationarity at all. A partial answer to this
  question is given in Feinberg's deficiency one algorithm. This
  algorithm can decide about the existence of multistationarity by
  analyzing a, potentially large, set of systems of linear
  inequalities that are independent of parameter values. 
  However, the deficiency one algorithm is limited to what its author
  calls \textbf{regular deficiency one networks}. Many
  realistic networks have a deficiency higher than one, thus the
  algorithm cannot be applied directly. In a previous publication it
  was suggested to analyze certain well defined subnetworks that are
  guaranteed to be of deficiency one. If these subnetworks are
  regular, then one can use the deficiency one algorithm to establish
  multistationarity. Realistic reaction networks, however, often lead
  to subnetworks that are irregular, especially if metabolic networks
  are considered. Here the special structure of the subnetworks
  is used to derive conditions for multistationarity. These conditions
  are independent of the regularity conditions required by the
  deficiency one algorithm. Thus, in particular, these conditions are
  applicable to irregular subnetworks.
\end{abstract}

\section{Introduction}

In Systems Biology there is a growing interest in the question, whether
or not a given mathematical model can admit more than one steady
state. In cell cycle regulation, for example, one can identify
different phases of the cell cycle (G1, S, G2 and M-phase) as
different stable steady states. The cycle itself can then considered as
a switching between these steady states. As parameter values (like
rate constants and total concentrations) are often unknown or subject
to a very high uncertainty due to measurement errors and and
difficult experimental conditions, one is often interested in the
question, whether or not a given mathematical model can, for some
conceivable parameter vector, exhibit multistationarity at all.

A partial answer to this question is given in Feinberg's chemical
reaction network theory, that links the ability of a mathematical
model to exhibit multistationarity to the structure of the underlying
biochemical reaction network
\cite{fein-006,fein-007,fein-016,fein-017}. The deficiency one
algorithm, in particular, can decide about the existence of
multistationarity by analyzing a, potentially large, set of systems of
linear inequalities that depend on the network structure alone, that
is, that are 
independent of parameter values. If any of these inequality
systems is feasible, then multistationarity is guaranteed and one can
compute steady states and rate constants from its solution set. If
all are infeasible, then multistationarity is impossible, for any
conceivable parameter vector (see, for example,
\cite{fein-006,fein-017}). Observe that, in particular, this algorithm
can also be used to prove that multistationarity is impossible.

However, the deficiency one algorithm is limited to what its author
calls \textbf{regular deficiency one networks} \cite{fein-017}. Many
realistic networks have a deficiency higher than one, thus the
algorithm cannot be applied directly. In \cite{fein-025,cc-flo-003} we
therefore suggested a way to circumvent this: instead of analyzing the
complete network we propose to analyze certain well defined
subnetworks that are guaranteed to be of deficiency one. If these
subnetworks are regular, then one can use the deficiency one algorithm
to establish multistationarity. If this is successful, then
\cite{fein-025} gives sufficient conditions that are computationally
simple to check to extend multistationarity from the subnetwork to
the overall network.

Realistic reaction networks, however, often lead to subnetworks that
are irregular, especially if metabolic networks are considered (see
e.g.\ \cite{subnet-003} for an analysis of the upper part of
glycolisis). Consequently, the deficiency one algorithm cannot be
applied to these subnetworks. If this irregularity is of a special
kind (termed $\emptyset$-irregularity in \cite{subnet-003}), then one can 
regularize the subnetwork and apply the deficiency one algorithm to
the resulting regularized subnetwork. It is then possible to extend
multistationarity -- so it exists -- to the overall network using the
aforementioned results of \cite{fein-025}.

Here we follow a different approach: instead of trying to regularize a
subnetwork, we use the special structure of the subnetworks defined in
\cite{fein-025} to derive conditions for multistationarity. These
conditions are independent of the regularity conditions required by
the deficiency one algorithm. Thus, in particular,
these conditions are applicable to irregular subnetworks. Of course it
is still possible to use the results of \cite{fein-025} to extend
multistationarity (once it can be established in the subnetwork).

\section{Notation}
\label{sec:Notation}

Consider the following (bio)chemical reaction network with $n=2$
species $A$ and $B$ and with $m=5$ complexes $A$, $0$, $B$, $A+B$
and $2\, A$ and $r=6$ reactions: 
\begin{displaymath}
  \xymatrix{
    A \ar @<.4ex> @{-^>} ^{\bf k_1}[r] & 0 \ar @{-^>} ^{\bf k_2}[l] 
    \ar @{-^>} @<.4ex> ^{\bf k_3}[r] & B \ar @{-^>} ^{\bf k_4}[l] \\
    A + B \ar @{-^>} @<.4ex> ^{\bf k_5}[r] & 2\, A \ar @{-^>} ^{\bf k_6}[l]
  }
\end{displaymath}
Let $x\in\R^n$ be the vector of species concentrations (e.g.\ let
$x_1$ be the concentration of $A$ and $x_2$ be the concentration of
$B$). By associating each concentration with the corresponding unit
vector $e_i$ of Euclidean space ($A$ with $e_1$ and $B$ with $e_2$ in
case of the example) one can define $m$ \lq complex\rq-vectors $y_i$
(in case of the example $y_1 = e_1$ for $A$, $y_2=0$, the
2-dimensional zero vector for the complex $0$, $y_3=e_2$ for $B$,
$y_4=e_1+e_2$ for $A+B$ and $y_5 = 2\, e_1$ for $2\, A$). Collect
these in a matrix $Y\in\R^{n\times m}$. For the example one obtains
\begin{align*}
  Y &= \left[
    \begin{array}{ccc}
      y_1 & \dots & y_5
    \end{array}
  \right] \\
  &= \left[
    \begin{array}{ccccc}
      1 & 0 & 0 & 1 & 2 \\
      0 & 0 & 1 & 1 & 0
    \end{array}
  \right].
\end{align*}
Let $I_a$ be the incidence matrix of the graph associated to the
reaction network in standard form as defined in
\cite{fein-016,fein-017}, that is a graph, where node labels are
unique. This means that one has $I_a\in\{-1,0-1\}^{m\times r}$. For
the example one obtains 
\begin{displaymath}
  I_a = \left[
    \begin{array}{rrrrrr}
      -1 &  1 &  0 &  0 &  0 &  0 \\
       1 & -1 & -1 &  1 &  0 &  0 \\
       0 &  0 &  1 & -1 &  0 &  0 \\
       0 &  0 &  0 &  0 & -1 &  1 \\
       0 &  0 &  0 &  0 &  1 & -1
    \end{array}
  \right]
\end{displaymath}
Finally let $k\in\Rp^r$ be the vector of rate constants, that is for
the example:
\begin{displaymath}
  k = \left(k_1,\, \ldots,\, k_6\right).
\end{displaymath}

The stoichiometric matrix $N$ is defined as the product
\begin{equation}
  \label{eq:factor_N}
  N := Y\, I_a
\end{equation}
of the matrix of complexes $Y$ and the incidence matrix of the
associated directed graph $I_a$. 
\begin{defi}[Reactant Complex, Educt, $\bar m$]
  \label{def:reactant_complex}
  A complex that has at least one outgoing edge is called a reactant
  complex. We use the symbol $\bar m\leq m$ to denote the number of
  reactant complexes. \\
  Let $y$ be a reactant complex. Then all species with
  indices contained in $\supp\brac{y}$ are called educts. 
\end{defi}
For simplicity we assume --  w.l.o.g. -- the following ordering of
complexes:
\begin{ass}[Complex Ordering]\label{ass:complex_ordering}
  Assume that the complexes are ordered such that the first $\bar m$
  complexes are reactant complexes.
\end{ass}
Under this assumption the mapping $\reac$ that associates every
reaction with its reactant complex has a particular simple form:
\begin{defi}[Mapping $\reac$]
  Let $\reac:\left\{1,\, \ldots,\, r\right\} \to \left\{1,\, \ldots,\,
    \bar m\right\}$ be defined as
  \begin{equation}
    \label{eq:def_reac}
    \reac\brac{j} = i\text{, $y_i$ is the tail of reaction $j$}.
  \end{equation}
\end{defi}
If mass-action kinetics is used, then the reaction rate
$v_i\brac{k,x}$ associated to the $i$-th reaction is given as the
monomial $v_i\brac{k,x} = k_i\, x^{y_{\reac\brac{i}}}$ (i.e.\ the
reaction rate $v_i\brac{k,x}$ is proportional to the product of the
educt concentrations). One obtains the following function
$\vel{k}{x}$:
\begin{defi}[$\vel{k}{x}$, $\Phi\brac{x}$, $\Psi\brac{x}$]\label{def:v_PHI_PSI}
  Using mass action kinetics, the vector of reaction rates is defined
  as
  \begin{subequations}
    \begin{align}
      \label{eq:def_v}
      \vel{k}{x} &:= \diag\brac{k}\, \Phi\brac{x}, \\
      \intertext{where}
      \Phi\brac{x} &:= \brac{x^{y_{\reac\brac{1}}},\, \ldots,\,
        x^{y_{\reac\brac{r}}}}^T
        \ . \\ 
      \intertext{Let $e_i$ denote the unit vectors of Euclidian
        $n$-space and define}
      \Pi &:= \rbrac{
        \begin{array}{c}
          e^T_{\reac\brac{1}} \\
          \vdots \\
          e^T_{\reac\brac{r}}
        \end{array}
      } \\
      \Psi\brac{x} &:= \brac{x^{y_i}}_{i=1,\, \ldots,\, \bar m}. \\
      \intertext{Note that this implies that}
      \Phi\brac{x} &= \Pi\, \Psi\brac{x} \\
      \intertext{and thus}
      \vel{k}{x} &= \diag\brac{k}\, \Phi\brac{x} \; = \;
      \diag\brac{k}\, \Pi\, \Psi\brac{x} \\
      \intertext{hold. Let}
      \hat Y &= \rbrac{y_i}_{i=1,\, \ldots,\, \bar m}
    \end{align}
    be a matrix having the exponents of $\Psi\brac{x}$ as column
    vectors (recall assumption~\ref{ass:complex_ordering} and note
    that this implies that $\hat Y$ contains the first $\bar m$
    columns of $Y$).
  \end{subequations}
\end{defi}
For the example one obtains
\begin{displaymath}
  \vel{k}{x} = \left( k_1\, x_1,\, k_2,\, k_3,\, k_4\, x_2,\, k_5,
    x_1\, x_2,\, k_6\, x_1^2 \right)^T.
\end{displaymath}
Observe that $\hat Y = Y$ and thus
\begin{displaymath}
  \Phi\brac{x} = \Psi\brac{x},
\end{displaymath}
in this case.
Then the following system of Ordinary Differential Equations (ODEs)
describes the dynamics of the species concentrations:
\begin{subequations}
  \begin{equation}
    \label{eq:def_ODEs}
    \dot x = N\, \vel{k}{x},
  \end{equation}
  If the stoichiometric matrix $N\in\R^{n\times r}$ does not have full
  row rank, the system is subject to \lq conservation
  relations\rq: let $s=\rank\brac{N}<n$, then there is a matrix
  $W\in\R^{n\times n-s}$ with $W^T\, N =0$ and
  \begin{equation} 
    \label{eq:def_con_rel}
    W^T\, x(t) = c
  \end{equation}
  along solutions $x(t)$ of (\ref{eq:def_ODEs}), cf.\ \cite{fein-032}.
\end{subequations}

As we are mainly interested in positive steady states, the pointed
polyhedral cone \cone{} is of particular interest. The symbol $E$ is
used to denote the unique (up to scalar multiplication) generators of
\cone. Let $p$ be the number of generators, if $p>1$, then
$E\in\R^{r\times p}$ is a matrix whose columns are the generators of
\cone, if $p=1$, then $E\in\R^r$ is a vector.
Define the set of all nonnegative vectors $x\in\Rnn^p$ such that $E\,
x$ is positive: 
\begin{equation}
  \label{eq:def_LAMBDA}
  \LAM := \left\{\, x\in\Rnn^p\; |\; E\, x >0\, \right\}.
\end{equation}

\section{Some remarks about positive steady states}
\label{sec:remarks-pos-ss}

The structure of (\ref{eq:def_ODEs}) motivates the following result
concerning positive steady states:
\begin{lem}[Existence of positive steady states]
  \label{lemma:existence-pos-ss}
  Consider a system of ODEs as in (\ref{eq:def_ODEs}), with
  stoichiometric matrix $N$ and let $E\in\Rnn^{r\times p}$ be the
  generator matrix of \cone. Let $k\in\Rp^r$ be given. Then \emph{the
    positive vector $a$} is a solution to the polynomial equation $N\,
  \vel{k}{a} =0$, if and only if there exists a vector
  $\lambda\in\LAM$ with
  \begin{equation}
    \label{eq:cond_k_E}
    k= \diag\brac{\Phi\brac{a^{-1}}}\, E\, \lambda.
  \end{equation}
\end{lem}
\begin{proof}
  Follows from the fact that $a>0$ and $k>0$ implies
  $\vel{k}{a}>0$. Thus $N\, \vel{k}{a}=0$ holds if and only if
  $\vel{k}{a} \in \cone$, that is, if and only if $\vel{k}{a} = E\,
  \lambda$, for some $\lambda\in\LAM$. As $\vel{k}{a}=\diag\brac{k}\,
  \Phi\brac{a}$ (\ref{eq:cond_k_E}) follows immediately.
\end{proof}
%
\begin{remark}
  If a positive steady state exists, then (\ref{eq:cond_k_E}) must
  hold. The condition (\ref{eq:cond_k_E}) can thus be used to
  constrain the set of rate constants that allow the existence (of at
  least one) positive steady state.
\end{remark}
\begin{remark}[Positive steady states]
  Consider a system of ODEs as in (\ref{eq:def_ODEs}) and let
  $E\in\Rnn^{r\times p}$ be the generator matrix of \cone. 
  The system has a positive steady state, iff $\cone\neq\emptyset$
  and the rows of $E$ are nonzero.
\end{remark}
\begin{remark}
  \label{coro:a_free}
  Consider a system of ODEs as in (\ref{eq:def_ODEs}) and let
  $E\in\Rnn^{r\times p}$ be the generator matrix of \cone{} and
  suppose that $E$ does not contain any zero rows. Then \emph{every
    positive vector} $a$ can be a steady state of (\ref{eq:def_ODEs}),
  by choosing $k$ as in (\ref{eq:cond_k_E}), where $\lambda\in\LAM$ is
  free and takes the role of the rate constants.
\end{remark}

\section{Subnetworks defined by stoichiometric generators}
\label{sec:SubNets}

In this section the following concepts from graph theory will be used
(two of them are standard definitions in graph theory, that stated
here merely for convenience, while the third, very common in CRNT, is
derived from those two):
\begin{defi}
  \begin{enumerate}
  \item[\cite{wiki:connected_component}] Connected component: the maximal connected subgraphs of a graph
  \item[\cite{wiki:connected_component}] Strongly connected component: a directed graph is called
    strongly connected if there is a path from each vertex in the
    graph to every other vertex. The strongly connected components
    (SCC) of a directed graph are its maximal strongly connected
    subgraphs.
  \item[\cite{fein-017}] Terminal strongly connected component: an scc that has no
    outgoing edge
  \end{enumerate}
\end{defi}
Next we recall some results concerning subnetworks
defined by stoichiometric generators
\begin{lem}[Properties of subnetworks defined by stoichiometric
  generators]
  \label{lemma:prop_subnet}
  For a subnetwork that is defined by a stoichiometric 
  generator $E$ the following properties hold:
  \begin{enumerate}[{(}a{)}]
  \item Graph of the network in normal form is a \emph{forest of
      trees}
  \item Terminal strongly connected components consist of a single
    node (complex)
  \item The deficiency of the network is one
  \item The deficiency of every connected component is zero
  \item If every connected component contains only one
    \textbf{terminal} strongly connected component, then the network
    is regular (in the sense of CRNT, cf.\cite{fein-005,fein-006}, for
    example)
  \item \label{item:unique_pos_ker}$\ker\brac{N} = \rbrac{E}$ 
  \end{enumerate}
\end{lem}
\begin{proof}\hfill
  \begin{itemize}
  \item[(a),(b)] Follow from the definition of the generators of \cone
  \item[(c)-(f)] A proof can be found in \cite{fein-025}
  \end{itemize}
\end{proof}
The following corollary is an immediate consequence of
Lemma~\ref{lemma:prop_subnet}, (\ref{item:unique_pos_ker}) and
Remark~\ref{coro:a_free}.
\begin{coro}
  Consider a biochemical reaction network that is defined by a
  stoichiometric generator. Then
  \begin{enumerate}[{(}i{)}]
  \item any positive vector $a$ is a steady state of
    (\ref{eq:def_ODEs}), if $k$ is chosen as in (\ref{eq:cond_k_E})
  \item for an arbitrary but fixed positive $a$, $k$ as in
    (\ref{eq:cond_k_E}) is fixed up to scalar multiplication (i.e.\
    the positive $\lambda$)
  \item (positive) scalar multiplication of $k$ corresponds to a
    \emph{time scaling} of the ODEs, thus one can -- w.l.o.g. --
    choose $\lambda=1$
  \end{enumerate}
\end{coro}
From here on we assume that the system has at least one positive
steady state, that is
\begin{ass}
  The vector of rate constants is given by 
  \begin{equation}\label{eq:k_subnet}
    k= \diag\brac{\Phi\brac{a^{-1}}}\, E
  \end{equation}
  for some $a\in\Rp^n$.
\end{ass}
Consider Lemma~\ref{lemma:existence-pos-ss} and especially the facts
that for networks defined by stoichiometric generators $E$ consists of
one (column) vector and that -- w.l.o.g. -- $\lambda=1$. Then the ODEs
(\ref{eq:def_ODEs}) are equivalent to
\begin{gather*}
  \dot x = N\, \vel{k}{x} = N\, \diag\brac{E}\, \diag\brac{\Phi\brac{a^{-1}}}\, 
    \Phi\brac{x} = N\, \diag\brac{E}\,
    \Phi\brac{\frac{x}{a}},
\end{gather*}
where $\Phi\brac{x} = \Pi\, \Psi\brac{x}$ (cf.\
Definition~\ref{def:v_PHI_PSI}). Thus 
\begin{equation}
  \label{eq:ODEs_subnet_pos_ss}
  \dot x = N\, \diag\brac{E}\, \Pi\, \Psi\brac{\frac{x}{a}} = N\,
  \diag\brac{E}\, \Pi\, \diag\brac{\Psi\brac{a^{-1}}}\, \Psi\brac{x}
\end{equation}
follows.
\begin{remark}
  Systems like (\ref{eq:ODEs_subnet_pos_ss}) are sometimes called
  \emph{generalized mass action systems}. For those systems reaction
  rates $v_i\brac{k,x}$ are still defined as monomials $k_i\,
  x^{y_{\reac\brac{{i}}}}$, however the exponent vector
  $y_{\reac\brac{{i}}}$ does not need to correspond to the reactant
  stoichiometry anymore.

  Further observe that for the special system defined in
  (\ref{eq:ODEs_subnet_pos_ss}) $\Psi\brac{a^{-1}}$ takes the role of
  the rate constants.
\end{remark}
To establish multistationarity we need to show the existence
of a second steady state $b\in\Rp^n$ with
\begin{displaymath}
  N\, \vel{k}{b} = 0,
\end{displaymath}
for the \textbf{same vector $k$}. That is $b$ must satisfy:
\begin{equation}
  \label{eq:condi_b}
  N\, \diag\brac{E}\, \Pi\, \diag\brac{\Psi\brac{a^{-1}}}\,
  \Psi\brac{b} = 0 
\end{equation}
Obviously $\ker\brac{N\, \diag\brac{E}\, \Pi} = \rbrac{\uu{1}_{\bar
    \rho}}$ (as N is the stoichiometric matrix of a subnetwork defined
by a stoichiometric generator; to see this recall that (i)
$\ker\brac{N}= \rbrac{E}$ and (ii) $\Pi$ has full column rank and
(iii) row vectors of $\Pi$ are unit vectors: thus $ \diag\brac{E}\,
\Pi\, \uu{1}_{\bar m} = E$). It follows that (\ref{eq:condi_b}) is
equivalent to (observe that $a$, $b>0$ implies $\Psi\brac{\frac{b}{a}}>0$):
\begin{displaymath}
  \Psi\brac{\frac{b}{a}} = \alpha\, \uu{1}_{\bar m}, \alpha>0
\end{displaymath}
Apply $\ln\brac{\cdot}$ to obtain the \emph{linear system}
\begin{subequations}
  \begin{equation}
    \label{eq:mu_eq}
    \hat{Y}^T\,  \mu = \ln\brac{\alpha}\, \uu{1}_{\bar m},
  \end{equation}
  where   
  \begin{equation}
    \label{eq:def_mu}
    \mu := \ln\frac{b}{a} = \brac{\ln\frac{b_1}{a_1}, \ldots,
      \ln\frac{b_n}{a_n}}^T.
  \end{equation}
\end{subequations}
The previous discussion motivates the following Lemma:
\begin{subequations}
  \begin{lem}[Parameterizing positive steady state solutions]
    \label{lemma:para_pos_ss}
    Consider the ODEs derived from a biochemical reaction network that is
    defined by a stoichiometric generator. Let
    \begin{equation}
      \label{eq:def_M}
      \mathcal{M} := \left\{ \mu\in\R^n\; |\; \exists \rho > 0\text{, such that
          $\hat{Y}^T\, \mu = \rho\, \uu{1}_{\bar m}$} \right\}.
    \end{equation}
    If $\mathcal{M} \neq \emptyset$, then $\mu\in\mathcal{M}$ and
    $a\in\Rp^{n}$ parameterize positive solutions of the polynomial
    equation $N\, \vel{k}{b} = 0$. Let $\mu\in\mathcal{M}$ and let
    \begin{align}
      \label{eq:k_mu_eq}
      k &= \lambda  \Phi\brac{a^{-1}},\, \lambda>0 \\
      \intertext{be the vector of rate constants. Further let}
      \label{eq:b_mu_eq}
      b &= \diag\brac{e^\mu}\, a.
    \end{align}
    Then $N\, \vel{k}{a} =0$ and $N\, \vel{k}{b} =0$ hold.
  \end{lem}
\end{subequations}
\begin{proof}
  Let $\mu\in\mathcal{M}$ let $k$ and $b$ be as in (\ref{eq:k_mu_eq}),
  (\ref{eq:b_mu_eq}), respectively. Observe that $N\, \vel{k}{a} = 0$
  follows from Lemma~\ref{lemma:existence-pos-ss}. Thus we have to
  show that $N\, \vel{k}{b} =0$ holds. To this end observe that
  \begin{align*}
    N\, \vel{k}{b} &= N\, \diag\brac{\lambda\, E}\,
    \diag\brac{\Phi\brac{a^{-1}}}\, \Phi\brac{b} \\
    &= \lambda\, N\,
    \diag\brac{E}\, \diag\brac{\Phi\brac{a^{-1}}}\,
    \Phi\brac{\diag\brac{e^\mu}\, a} \\ 
    &=  \lambda\, N\, \diag\brac{E}\, \diag\brac{\Phi\brac{a^{-1}}}\,
    \diag\brac{\Phi\brac{a}}\, \Phi\brac{e^\mu} \\
    &= \lambda\, N\, \diag\brac{E}\, \Pi\, \Psi\brac{e^\mu} \\
    &= \lambda\, N\, \diag\brac{E}\, \Pi\, e^{\hat{Y}^T\, \mu} \\ 
    &= \lambda\, N\, \diag\brac{E}\, \Pi\, e^\rho\, \underline{1}_{m_1}
    = 0.
  \end{align*}
\end{proof}
\begin{remark}\label{rem:inifite_ss}
  If either (i) $\underline{1}_{\bar m} \in \rbrac{\hat{Y}^T}$ or (ii)
  $\ker\brac{\hat{Y}^T}$ is nontrivial (i.e.\ $\rank\brac{\hat{Y}} <
  n$), or both, then a \emph{fixed vector $a\in\Rp^n$} together with
  $k$ as in (\ref{eq:k_mu_eq}) defines an \emph{infinite} set of
  positive steady states $b\in\Rp^n$. To see
  this assume that (i), (ii) or both hold. Then (\ref{eq:mu_eq}) is
  solvable and the solution set
  \begin{displaymath}
    \mathcal{M} := \left\{ \mu\in\R^n\; |\; \exists \rho > 0\text{, such that
        $\hat{Y}^T\, \mu = \rho\, \uu{1}_{\bar m}$} \right\}
  \end{displaymath}
  defines a linear subspace, that is there exists a matrix $M$ and a
  vector $\kappa$ of appropriate dimensions, such that
  \begin{displaymath}
    \mu\in\mathcal{M} \Leftrightarrow \mu = M\, \kappa
  \end{displaymath}
  holds. As every $\mu\in\mathcal{M}$ defines a positive $b$ and as a
  linear vector space contains infinitely many elements $\mu$, there
  exists, \emph{for a fixed $a\in\Rp^n$} infinitely many $b\in\Rp^n$ with
  \begin{equation}
    \label{eq:def_b_a_kap}
    b_a\brac{\kappa} = \diag\brac{e^{M\, \kappa}}\, a.
  \end{equation}
\end{remark}

\section{Subnetwork multistationarity}
\label{sec:subnet_multistat}

If the conditions of Remark~\ref{rem:inifite_ss} hold, then there
exists an infinite set of positive steady states, even for a given
$a\in\Rp^n$ (recall that $\Psi\brac{a^{-1}}$ takes the role of the rate
constants). Fix $a\in\Rp$ and thus $k =
\diag\brac{\Phi\brac{a^{-1}}}\, E = \diag\brac{E}\, \Pi\,
\Psi\brac{a^{-1}}$. Then all $b_a\brac{\kappa}$ as defined in
(\ref{eq:def_b_a_kap}) are steady states. However, for a given initial
condition $x_0\in\Rp^n$ the system \lq sees\rq{} only a subset of set
of positive steady states.

To see this recall the ODEs derived from a biochemical reaction
network $\dot x = N\, \vel{k}{x}$ and let $s :=\rank\brac{N}
< n$. Let $S$, $W$ be \emph{orthonormal bases} of $\rbrac{N} =:
\mathcal{S}$ and $\mathcal{S}^\perp$, respectively. Similar to
\cite{fein-032}, introduce the transformation 
\begin{displaymath}
  y=S^T\, x,\; z= W^T\, x\; \text{and $x=\xi\brac{x,y} = S\, y + W\,
    z$}. 
\end{displaymath}
In the new coordinates the ODEs read as
\begin{align*}
  \dot y &= S^T\, \dot x = S^T\, N\, \diag\brac{E}\, \Pi\,
  \diag\brac{\Psi\brac{a^{-1}}}\,
  \Psi\brac{\xi\brac{y,z}} \\ 
  \dot z &= W^T\, \dot x = 0
\end{align*}
showing the invariance of $\mathcal{S}$. Let
$x\brac{0}=x_0\in\Rp^n$, then the solution $x\brac{t}$ is given by 
\begin{displaymath}
  x\brac{t} = x_0 + \int_0^t\,  N\, \diag\brac{E}\, \Pi\,
  \diag\brac{\Psi\brac{a^{-1}}}\, \Psi\brac{x\brac{\tau}} d\tau.
\end{displaymath}
For the new coordinates one obtains
\begin{align*}
  y\brac{t} &= S^T\, x_0 + S^T\, \int_0^t\,  N\, \diag\brac{E}\, \Pi\,
  \diag\brac{\Psi\brac{a^{-1}}}\, \Psi\brac{x\brac{\tau}} d\tau \\
  z\brac{t} &= \tilde{W}^T\, x_0 = const.
\end{align*}
that is solutions are confined to parallel translates of
$\mathcal{S}$. Thus, for a given initial condition, the system
\lq sees\rq{} only those positive steady states that are in the
intersection of $b_a\brac{\kappa}$ and $\mathcal{S}$. Observe that
$\tilde{\mathcal{S}} := \rbrac{N\, \diag\brac{E}\, \Pi} \subseteq
\mathcal{S}$ and that, by a similar argument, solutions are confined
to parallel translates of $\tilde{\mathcal{S}}$. This motivates the
following definition of multistationarity with respect to a linear
subspace as introduced in \cite{cc-flo-003}:
\begin{defi}
  \label{defi-multi} 
  Given a subspace ${\cal V}\subset \R^n$,
  the system $\dot x =  N\, v(k,x)$ from (\ref{eq:def_ODEs}) with
  stoichiometric subspace ${\cal S}=\im{(N)}$ is said to exhibit
  ${\cal V}$-multistationarity if and only if there exist 
  a positive vector $k \in
  \Rp^r$ and
  at least two distinct
  positive vectors $a$, $b \in \Rp^n$ with
  \begin{subequations}
    \begin{align}
      \label{eq:multi_ode_0}
      N\, v(k,a) \ = 0, \quad N\, v(k,b) \ = 0 , \\
      \label{eq:multi_ode_1}
      b-a \in {\cal V}.
    \end{align}
  \end{subequations}
\end{defi}
\begin{remark}
  Note that if $\mathcal{V} = \mathcal{S}$, then
  Definition~\ref{defi-multi} is equivalent to the familiar definition
  of multistationarity in Chemical Engineering and
  especially in CRNT as defined, for example, in
  \cite{fein-006,fein-017}.
\end{remark}
\begin{remark}
  Note that for subnetworks defined by stoichiometric generators two
  linear subspaces are of particular interest: $\rbrac{N\,
    \diag\brac{E}\, \Pi}$ and the image of stoichiometric matrix of
  the overall network $\rbrac{\hat N}$. Multistationarity with respect
  to $\rbrac{N\, \diag\brac{E}\, \Pi}$ means that the subnetwork can
  exhibit multistationarity, if it is considered in isolation, while
  multistationarity with respect to $\rbrac{\hat N}$ means that the
  subnetwork as part of the larger network can exhibit
  multistationarity. Thus, if a subnetwork exhibits $\rbrac{\hat
    N}$-multistationarity, but not $\rbrac{N\, \diag\brac{E}\, 
    \Pi}$-multistationarity, then this subnetwork can give rise to
  multistationarity for the overall network, even though, in
  isolation, it does not exhibit multistationarity.
\end{remark}
\begin{remark}
  As an illustration consider the network in
  Fig.\ref{fig:exa_net}. For this network the matrix $N$ is given by
  \begin{subequations}
    \begin{align}
      \label{eq:N_exa}
      N &= \rbrac{
        \begin{array}{rrrr}
          1 & 0 & -1 & 0\\
          0 & 1 & -1 & 0\\
          0 & 0 & 1 & -1
        \end{array}
      } \\
      \intertext{the unique generator of $\cone$ is given by}
      \label{eq:E_exa}
      E &= \left(\, 1,\, 1,\, 1,\, 1\, \right)^T.
      \intertext{For $\Phi\brac{x}$ one obtains}
      \Phi\brac{x} &= \left( 1,\, 1,\, x_3,\, x_1\, x_2 \right)^T \\
      \intertext{and therefore}
      \Pi &=\left[
        \begin{array}{ccc}
          1 & 0 & 0 \\
          1 & 0 & 0 \\
          0 & 1 & 0 \\
          0 & 0 & 1
        \end{array}
      \right] & \Psi\brac{x} &= \left( 1,\, x_3,\, x_1\, x_2 \right)^T.
      \intertext{Thus one obtains}
      \label{eq:k_exa}
      k &= \left(\, 1,\, 1,\, \frac{1}{a_1\, a_2},\, \frac{1}{a_3}\,
      \right)^T 
      \intertext{for arbitrary $a>0$. Observe that for this example
        $N\, \diag\brac{E}\, \Pi = N\, \Pi$. It is straightforward to
        verify  that all points on the following one-dimensional curve
        are steady states (parameterized by $p>0$):}
      \label{eq:ss_exa}
      x_s\brac{p} &= \left(\, p,\, \frac{1}{p},\, 1\, \right).
    \end{align}
  \end{subequations}
  As $N$ has full row rank, the left kernel is trivial, that is, it is
  spanned by $W=0$, the three-dimensional zero vector. Pick two
  distinct real numbers $p_1>$, $p_2>0$. Then $x_s\brac{p_1}$,
  $x_s\brac{p_2}$ are steady states that satisfy $W^T\,
  \brac{x_s\brac{p_1}-  x_s\brac{p_2}}=0$ (i.e.\ . $x_s\brac{p_1}-
  x_s\brac{p_2} \in \rbrac{N}$. Thus, according to our definition, the
  system exhibits $\rbrac{N}$-multistationarity.

  It fails to exhibit $\rbrac{N\, \Pi}$-multistationarity, as for a
  given initial condition $x_0>0$, all trajectories converge to a
  unique steady state. That is due to the fact that
  \begin{displaymath}
    N\, \Pi = \rbrac{
     \begin{array}{rrr}
        1 & -1 & 0\\
        1 & -1 & 0\\
        0 & 1 & -1
      \end{array}
    }
  \end{displaymath}
  has a nontrivial left kernel $\tilde W=\left(\, 1,\, -1,\, 0\,
  \right)^T$. Thus, for a given initial condition, a trajectory is
  confined to an affine linear subspace perpendicular to $\tilde
  W$. And all trajectories starting in particular affine linear
  subspace converge to the same steady state (demonstrated numerically
  for $a=\uu{1}$ in Fig.\ref{fig:3d_simu} and
  \ref{fig:x1_x2_plane}). Note that, using $k$ as in (\ref{eq:k_exa})
  the ODEs are equivalent to a system of ODEs derived from the network
  displayed in Fig.\ref{fig:transformed_exa_net}, a weakly reversible
  deficiency zero network. From the Deficiency Zero Theorem
  \cite{fein-005} follows that this network has a unique,
  asymptotically stable positive steady state -- relative to a given
  initial condition.

  Further note that the network in Fig.~\ref{fig:exa_net} is irregular
  in the sense of CRNT (cf.\ \cite{fein-006}) and that it fails to
  exhibit $\rbrac{N\, \Pi}$-multistationarity, while it exhibits
  $\rbrac{N}$-multistationarity.
\end{remark}
\begin{figure}
  \subfigure[Example network]{\label{fig:exa_net}
    \begin{minipage}[c]{0.5\linewidth}
      \begin{displaymath}
        \xymatrix{
          & & & {\bf A} \\
          {\bf A + B} \ar [r] ^{\bf k_4} & {\bf C} \ar [r] ^{\bf k_3}
          & 0 \ar [ur] ^{\bf k_1} \ar [dr] ^{\bf k_2}& \\
          & & & {\bf B}
        }
      \end{displaymath}
    \end{minipage}
    }
    \subfigure[Transformed network]{\label{fig:transformed_exa_net}
      \begin{minipage}[c]{0.5\linewidth}
        \begin{displaymath}
          \xymatrix{
            0 \ar [r] ^{\bf k^*} & {\bf A + B} \ar [r] ^{\bf k_4} & {\bf
              C} \ar `d[ll] `[ll] ^{\bf k_3}  [ll]  \\
          }
      \end{displaymath}
    \end{minipage}
  }\\
  \subfigure[Simulation for selected initial conditions]{\label{fig:3d_simu}
    \includegraphics[width=6cm]{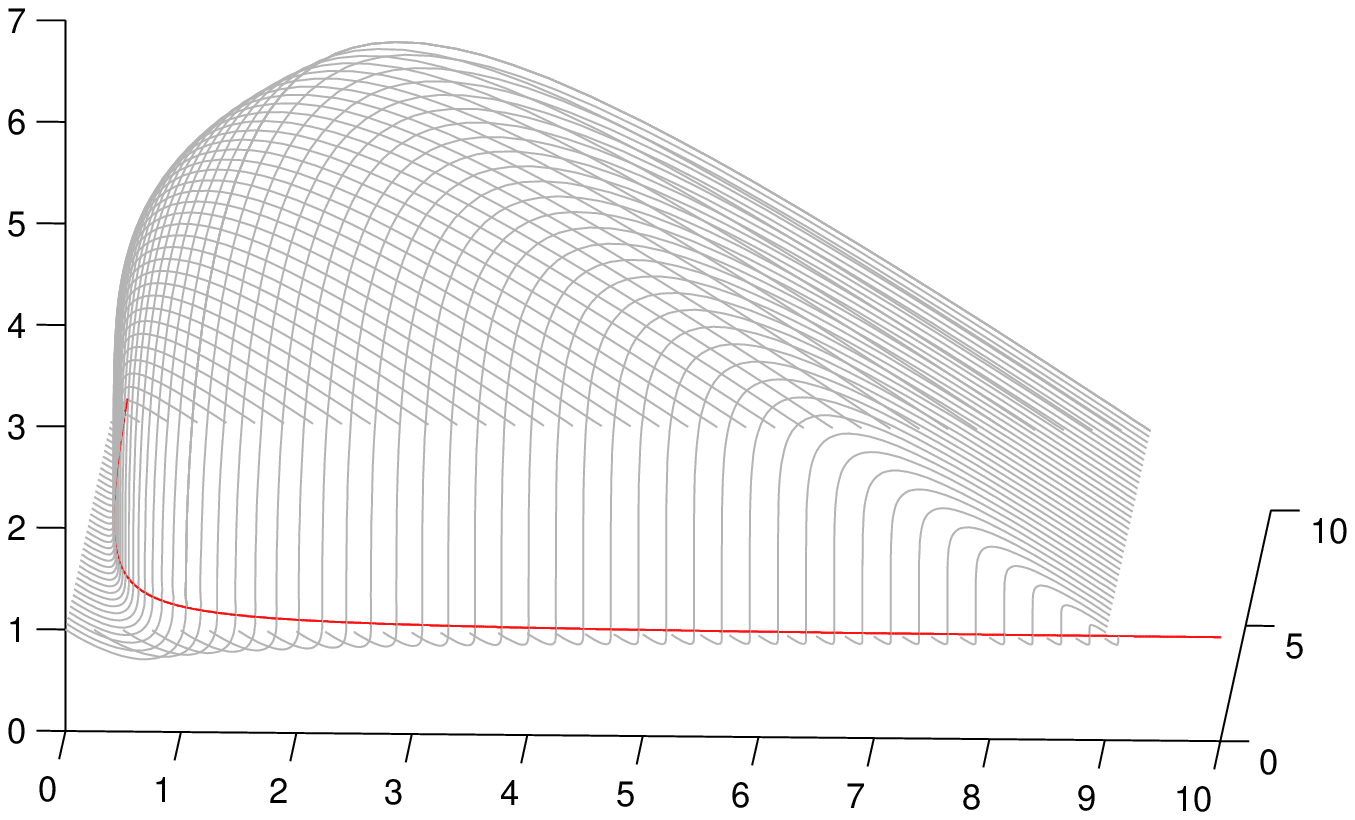}} 
  \subfigure[Projection in the $x_1$-$x_2$-plane]{\label{fig:x1_x2_plane}
    \includegraphics[width=6cm]{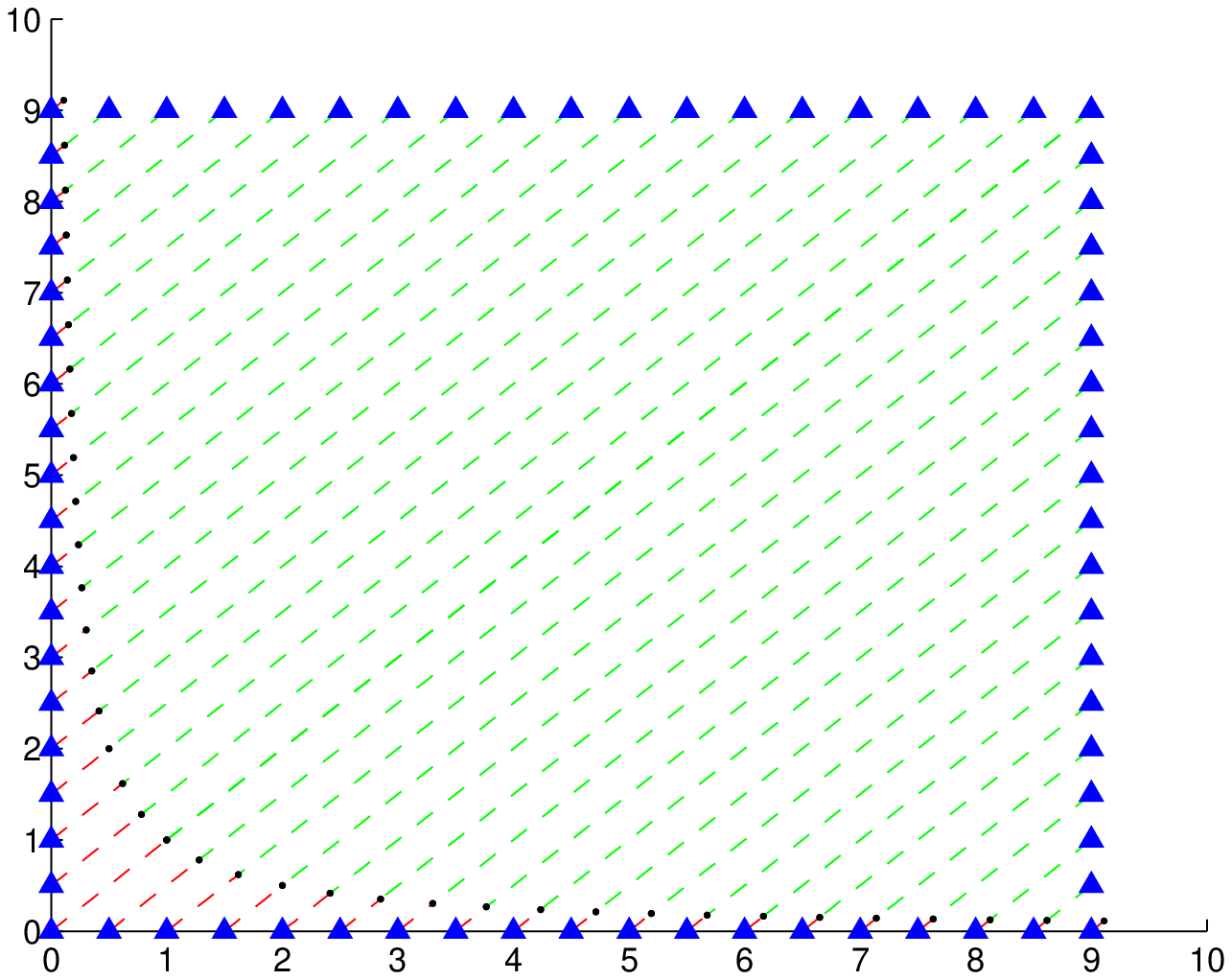}} 
\end{figure}

\section{Establishing multistationarity for subnetworks}

Consider a biochemical reaction network defined by a stoichiometric
generator. From Section~\ref{sec:SubNets} it is known, that for a
given $a\in\Rp^n$ all points $b_a\brac{\kappa}$ as defined in
(\ref{eq:def_b_a_kap}) are steady states. Moreover, the set
$\mathcal{M}$, a linear subspace, as defined in (\ref{eq:def_M})
contains all $\mu=\ln b_a\brac{\kappa} - \ln a$. From
Section~\ref{sec:subnet_multistat} it is known, that
$\mathcal{V}$-multistationarity requires

\begin{align*}
  b_a\brac{\kappa} - a &\in \mathcal{V} \\
  \ln  b_a\brac{\kappa} - \ln a &\in \mathcal{M}
\end{align*}
To this end a result from \cite{fein-024} can be used. To state the
result, let \sig{u} denote the sign pattern of the vector $u\in\R^n$. Then
$v=\sig{u}$ is a vector with entries $v_i\in\{+,-,0\}$ depending on
whether $u_i>0$, $u_i<0$ or $u_i=0$, respectively.
\begin{lem}[cf. \cite{fein-024}]
  \label{lemma:sigma_mu}
  Let $M_1\subseteq \R^n$ and $M_2\subseteq \R^n$ be two nontrivial
  subsets of $\R^n$ and define $M_3:=\big\{
  \brac{m_1,m_2} \in M_1\times M_2 \big|$ $\sig{m_1} = \sig{m_2}
  \big\}$ as the set of all ordered pairs $(m_1,m_2)$ of elements
  $m_1\in M_1$ and $m_2\in M_2$ with the same sign pattern. Two
  positive vectors $p$ and $q$ with    $\ln q - \ln p \in M_1$ and 
  $q-p \in M_2$ exist, if and only if $M_3 \neq \emptyset$. Then $p$
  and $q$ are given by
  \begin{align}
    \label{eq:p_m2_exp_m1}
    \left(p_i\right)_{i=1,\, \ldots,\, n} &= 
    \begin{cases}
      \frac{m_{2i}}{e^{m_{1i}}-1}\text{, if $m_{1i} \neq 0$} \\
      \bar{p}_i>0\text{, if $m_{1i} = 0$,}
    \end{cases}
    \intertext{where $\bar{p}_i$ denotes an arbitrary positive number
      and}
    \label{eq:q_exp_m1_p}
    \left(q_i\right)_{i=1,\, \ldots,\, n} &= e^{m_{1i}}\, p_i.
  \end{align}
\end{lem}
Thus -- using $b_a\brac{\kappa}$ instead of $q$, $a$ instead of $p$ and
$\mathcal{M}$ instead of $M_1$, $\mathcal{V}$ instead of $M_2$ -- all
one has to do is to find a vector $\mu\in\mathcal{M}$ and a vector
$s\in\mathcal{V}$ with $\sig{\mu} = \sig{s}$ to establish
$\mathcal{V}$-multistationarity.


\end{document}